\documentclass[
]{no-ceurart}

\usepackage{graphbox}

\usepackage{algorithm}
\usepackage[noend]{algpseudocode}
\algrenewcommand\algorithmicrequire{\textbf{Input:}}
\algrenewcommand\algorithmicensure{\textbf{Output:}}

\usepackage{amsthm}
\usepackage{cleveref}
\newtheorem{theorem}{Theorem}
\newtheorem{lemma}[theorem]{Lemma}

\newtheorem{definition}[theorem]{Definition}

\usepackage{amsmath}
\usepackage{amssymb}
\usepackage{stmaryrd}
\usepackage{mathtools}
\usepackage{thm-restate}
\usepackage{csquotes}
\usepackage{xspace}
\usepackage{xcolor}

\usepackage{todonotes}
\usepackage{proofread}

\usepackage{macros}

\newcommand{\inAppendix}{in the appendix}

\begin{document}

\copyrightyear{2024}
\copyrightclause{Copyright for this paper by its authors.
  Use permitted under Creative Commons License Attribution 4.0
  International (CC BY 4.0).}

\conference{This is the extended version of a paper presented at DL 2024: 37th International Workshop on Description Logics.}

\title{The Precise Complexity of Reasoning in ALC with Omega-Admissible Concrete Domains (Extended Version)}[The Precise Complexity of Reasoning in $\mathcal{ALC}$ with $\omega$-Admissible Concrete Domains (Extended Version)]
%

\author[1]{Stefan {Borgwardt}}[%
  orcid=0000-0003-0924-8478,
  email=stefan.borgwardt@tu-dresden.de,
  url=https://lat.inf.tu-dresden.de/~stefborg,
]
\author[1,2]{Filippo {De Bortoli}}[%
  orcid=0000-0002-8623-6465,
  email=filippo.de_bortoli@tu-dresden.de,
  url=https://lat.inf.tu-dresden.de/~debortoli,
]
\author[3]{Patrick {Koopmann}}[%
  orcid=0000-0001-5999-2583,
  email=p.k.koopmann@vu.nl,
  url=https://pkoopmann.github.io,
]

\address[1]{%
  TU Dresden, Institute of Theoretical Computer Science \\
  Dresden, Germany
}
\address[2]{%
  Center for Scalable Data Analytics and Artificial Intelligence (ScaDS.AI) \\
  Dresden/Leipzig, Germany
}
\address[3]{%
  Vrije Universiteit Amsterdam, Department of Computer Science \\
  Amsterdam, The Netherlands
}

\begin{abstract}
  Concrete domains have been introduced in the context of Description Logics to allow references to qualitative and quantitative values.
  In particular, the class of $\omega$-admissible concrete domains, which includes Allen's interval algebra, the region connection calculus (RCC8), and the rational numbers with ordering and equality, has been shown to yield extensions of \ALC for which concept satisfiability w.r.t.\ a general TBox is decidable.
  In this paper, we present an algorithm based on type elimination and use it to show that deciding the consistency of an \ALCD ontology is \ExpTime-complete if the concrete domain $\cDom$ is $\omega$-admissible and its constraint satisfaction problem is decidable in exponential time.
  While this allows us to reason with concept and role assertions, we also investigate \emph{feature assertions} $f(a,c)$ that can specify a constant~$c$ as the value of a feature~$f$ for an individual~$a$.
  We show that, under conditions satisfied by all known $\omega$-admissible domains, we can add feature assertions without affecting the complexity.
\end{abstract}

\begin{keywords}
  Description Logics \sep
  Concrete Domains \sep
  Reasoning \sep
  Complexity \sep
  Type Elimination
\end{keywords}

\maketitle

\section{Introduction}
\label{sec:intro}

Reasoning about numerical attributes of objects is a core requirement of many applications.
For this reason, \emph{Description Logics (DLs)} that integrate reasoning over an abstract domain of knowledge with references to values drawn from a \emph{concrete domain} have already been investigated for over 30 years~\cite{BaHa91,DBLP:phd/dnb/Lutz02,DBLP:conf/aiml/Lutz02,DBLP:conf/ijcai/BaaderBL05,LuMi07,DBLP:conf/ecai/ArtaleRK12,CaTu16,LaOrSi20,BaRy22,DeQu23a}.
In this setting, DLs are extended by \emph{concrete features} that are interpreted as partial functions mapping abstract domain elements to concrete values, such as the feature $\ex{diastolic}$ that describes the diastolic blood pressure of a person.
Using \emph{concrete domain restrictions} we can then describe constraints over these features and state, e.g.\ that all patients have a diastolic blood pressure that is lower than their systolic blood pressure, by adding the concept inclusion $\ex{Patient} \sqsubseteq\exists\ex{diastolic},\ex{systolic}.{<}$ to the ontology.
More interestingly, we can use \emph{feature paths} to compare the feature values of different individuals: the inclusion $\top\sqsubseteq\forall\ex{hasChild}\,\ex{age},\ex{age}.{<}$ states that children are always younger than their parents.

Unfortunately, dealing with feature paths in DLs is very challenging and often leads to undecidability~\cite{DBLP:phd/dnb/Lutz02,BaRy20,BaRy22}, which is why restrictive conditions on the concrete domains were introduced to regain decidability.
An approach inspired from research on constraint satisfaction is based on \emph{$\omega$-admissibility}~\cite{LuMi07}, which requires certain compositionality properties for finite and countable sets of constraints.
This enables the composition of full models from local solutions of sets of concrete domain constraints.
Being quite restrictive, at first only two examples of $\omega$-admissible concrete domains were known, namely Allen's interval algebra and the region connection calculus (RCC8).
Recently, the conditions of $\omega$-admissibility were investigated in more detail, and it was shown that $\omega$-admissible concrete domains can be obtained from \emph{finitely bounded homogeneous structures}~\cite{BaRy22}.
This class includes the rational numbers $\stru{Q} := (\mathbb{Q},<,=,>)$ with ordering and equality, and is closed under union and certain types of product~\cite{BaRy22}.

Despite these extensive investigations, the precise complexity of reasoning using description logics with concrete domains has been established only for a few special cases so far~\cite{DBLP:phd/dnb/Lutz02,CaTu16,LaOrSi20,DeQu23a}.
These results have been obtained by automata-based techniques, and yield decidability or tight complexity bounds for cases where the concrete domain is not $\omega$-admissible, such as the ordered integers $(\mathbb{Z},<,=,>)$~\cite{LaOrSi20} or the set of strings over a finite alphabet with a prefix relation~\cite{DeQu23a}.
In contrast, the goal of this paper is to derive a tight complexity bound for ontology consistency in the extension of \ALC with any $\omega$-admissible concrete domain.
Using an algorithm based on type elimination, we establish that this decision problem is \ExpTime-complete under the assumption that the concrete domain $\cDom$ is $\omega$-admissible and that its \emph{constraint satisfaction problem} (CSP) is decidable in exponential time.
The latter condition holds for all finitely bounded homogeneous structures (including the three examples from above), since their CSP is guaranteed to be in NP and in some cases, e.g. for $\stru{Q}$, is decidable in polynomial time.

Additionally, we consider \emph{predicate} and \emph{feature assertions}.
These allow us to constrain values of features associated to named individuals, e.g.\ to write ${<}(\ex{age}(\ex{mary}),\ex{age}(\ex{bob}))$ to state that Mary is younger than Bob, or to specify a constant value of a feature for a named individual, e.g.\ using $\ex{systolic}(\ex{mary}, 122)$ to state that Mary's systolic blood pressure equals~$122$.
We prove that consistency in the presence of predicate assertions can always be reduced in polynomial time to consistency without predicate assertions.
For feature assertions, we show that this holds if we consider \emph{homogeneous} $\omega$-admissible concrete domains, which again is the case for the three concrete domains from above.
Finally, we show that supporting feature assertions is equivalent to supporting singleton predicates of the form~$=_c$ for a constant~$c$ in the concrete domain.

Omitted proofs of the results presented in the text can be found \inAppendix.
\section{Preliminaries}
\label{sec:preliminaries}

\mypar{Concrete Domains}
As usual for DLs, we adopt the term \emph{concrete domain} to refer to a relational structure $\cDom=(D,P_1^D,P_2^D,\dots)$ over a non-empty, countable relational signature $\{P_1,P_2,\dots\}$, where $D$ is a non-empty set, and each predicate~$P$ has an associated arity $k \in \mathbb{N}$ and is interpreted by a relation $P^D \subseteq D^k$.
An example is the structure $\stru{Q} := (\mathbb{Q},{<},{=},{>})$ over the rational numbers~$\mathbb{Q}$ with standard binary order and equality relations.
Given a countably infinite set~$V$ of variables, a \emph{constraint system} over $V$ is a set~$\cs$ of \emph{constraints} $P(v_1,\dotsc,v_k)$, where $v_1,\dotsc,v_k \in V$ and $P$ is a $k$-ary predicate of $\cDom$.
We denote by $V(\cs)$ the set of variables that occur in~$\cs$.
The constraint system $\cs$ is \emph{satisfiable} if there is a \emph{homomorphism} $h \colon V(\cs) \to D$ that satisfies every constraint in $\cs$, i.e.\ $P(v_1,\dotsc,v_k) \in \cs$ implies $(h(v_1),\dotsc,h(v_k)) \in P^D$.
We then call $h$ a \emph{solution} of $\cs$.
The \emph{constraint satisfaction problem} for $\cDom$, denoted $\CSP(\cDom)$, is the decision problem asking whether a finite constraint system~$\cs$ over~$\cDom$ is satisfiable.
The problem $\CSP(\stru{Q})$ is in P, since satisfiability can be reduced to ${<}$-cycle detection; for example, the 3-clique $\{ x_1 < x_2,\ x_2 < x_3,\ x_3 < x_1 \}$ (using infix notation for~$<$) is unsatisfiable over~$\stru{Q}$.

To ensure that reasoning in the extension of \ALC with concrete domain restrictions is decidable, we impose further properties on~$\cDom$ regarding its relations and the compositionality of its CSP for finite and countable constraint systems.
We say that $\cDom$ is a \emph{patchwork} if it satisfies the following conditions:\footnote{\cite{LuMi07} originally used only \JEPD (jointly exhaustive, pairwise disjoint) and \AP (amalgamation property). \JD (jointly diagonal) was later added by~\cite{BaRy22}.}
\begin{description}
  \item[\JEPD] for all $k \ge 1$, either $\cDom$ has no $k$-ary relation, or $D^k$ is partitioned by all $k$-ary relations;
  \item[\JD] there is a quantifier-free, equality-free first-order formula over the signature of $\cDom$ that defines the equality relation~$=$ between two elements of $\cDom$;
  \item[\AP] if $\ocs$, $\cs$ are constraint systems and $P(v_1,\dotsc,v_k) \in \ocs \;\text{iff}\; P(v_1,\dotsc,v_k) \in \cs$ holds for all $v_1,\dotsc,v_k \in V(\ocs)\cap V(\cs)$ and all $k$-ary predicates $P$ over $\cDom$, then $\ocs$ and $\cs$ are satisfiable iff $\ocs \cup \cs$ is satisfiable.
\end{description}
If $\cDom$ is a patchwork, we call a constraint system $\cs$ \emph{complete} if, for all $k \in \mathbb{N}$, either $\cDom$ has no $k$-ary predicates, or for all $v_1,\dotsc,v_k \in V(\cs)$ there is exactly one $k$-ary predicate~$P$ over $\cDom$ such that $P(v_1,\dotsc,v_k) \in \cs$.
The concrete domain $\cDom$ is \emph{homomorphism $\omega$-compact} if a countable constraint system $\cs$ over~$\cDom$ is satisfiable whenever every finite constraint system $\cs' \subseteq \cs$ is satisfiable.
We now introduce our definition of \expadm concrete domains, which differs from the definitions of $\omega$-admissibility~\cite{LuMi07,BaRy22} in that we require $\CSP(\cDom)$ to be decidable in exponential time instead of simply decidable.

\begin{definition}
  \label{dfn:omega-admissible}
  A concrete domain $\cDom$ is \emph{\expadm}
  if
  \begin{itemize}
    \item $\cDom$ has a finite signature,
    \item $\cDom$ is a patchwork,
    \item $\cDom$ is homomorphism $\omega$-compact, and
    \item $\CSP(\cDom)$ is decidable in exponential time.
  \end{itemize}
\end{definition}
Requiring the signature of $\cDom$ to be finite is necessary to ensure decidability of \ALCD~\cite{BaRy22}.
In~\Cref{sec:assertions}, we will also consider concrete domains~$\cDom$ that are \emph{homogeneous}, that is, every isomorphism between finite substructures of~$\cDom$ can be extended to an isomorphism from~$\cDom$ to itself.
All properties we have defined here are satisfied by the three examples of Allen's interval relations, RCC8, and $\stru{Q}$, as they are \emph{finitely bounded} homogeneous structures, which are $\omega$-admissible and have CSPs whose complexity is at most NP (see~\cite{BaRy22} for details).

\mypar{Ontologies with Concrete Domain Constraints}
We assume the reader to be familiar with the standard description logic \ALC~\cite{DBLP:books/daglib/0041477}.
To use a concrete domain~$\cDom$ in \ALC axioms, we introduce a set of \emph{concrete features}~$\NF$, where each $f\in \NF$ is interpreted as a \emph{partial} function $f^\Imc\colon\Delta^\Imc\to D$ in an interpretation~\Imc.
A \emph{feature path} is of the form $r_1\dots r_nf$, where $r_1,\dots,r_n\in \NR$ are role names and $f\in \NF$.
The semantics of such a path is given by a function
\[
  (r_1\dots r_nf)^\Imc(d) \coloneqq
  \big\{ f^\Imc(e) \mid (d,e)\in r_1^\Imc\circ\dots\circ r_n^\Imc \;\text{and}\; f^\Imc(e) \;\text{is defined} \big\}
\]
that assigns to each domain element the set of all $f$-values of elements~$e$ reachable via the role chain $r_1\dots r_n$ (which may be empty if there is no such~$e$, or if $f^\Imc(e)$ is always undefined).
In a slight abuse of notation, we allow the case where $n=0$, i.e.\ $f$ can be seen as both a feature and a feature path, with slightly different, but equivalent semantics (that is, a partial function vs.\ a set-valued function that may produce either a singleton set or the empty set).
\emph{\ALCD concepts} are defined similarly to \ALC concepts, but can additionally use the following concept constructor: A \emph{concrete domain restriction} (or simply \emph{CD-restriction}) is of the form $\exists p_1,\dots,p_k.P$ or $\forall p_1,\dots,p_k.P$, where $p_1,\dots,p_k$ are feature paths and $P$ is a $k$-ary predicate, with the semantics
\begin{align*}
  (\exists p_1,\dotsc,p_k.P)^\Imc & \coloneqq \big\{ d\in\Delta^\Imc \mid \text{there are}\; c_i\in p_i^\Imc(d) \;\text{for}\; i = 1,\dotsc,k \;\text{s.t.}\; (c_1,\dots,c_k)\in P^D \big\}, \\
  (\forall p_1,\dots,p_k.P)^\Imc  & \coloneqq \big\{ d\in\Delta^\Imc \mid \text{if}\; c_i\in p_i^\Imc(d) \;\text{for}\; i = 1,\dotsc,k, \;\text{then}\; (c_1,\dots,c_k)\in P^D\big\}.
\end{align*}
An \emph{\ALCD ontology} $\Omc=\Amc\cup\Tmc$ consists of an \emph{ABox}~\Amc and a \emph{TBox}~\Tmc over \ALCD concepts. 

Note that we explicitly consider \ALCD and not \ALCFD, i.e.\ we do not allow functional roles, in contrast to other definitions from the literature~\cite{BaHa91,LuMi07}. In the literature, feature paths are either restricted to contain only functional roles, or to have a length of at most 2, (compare with~\cite{DBLP:phd/dnb/Lutz02,LuMi07}). Consequently, we restrict ourselves to feature paths of length $\le 2$, that is, assume that they are of the form $f$ or $rf$.
We are not aware of any work that considers feature paths over non-functional roles of length longer than 2, and leave the investigation of this case for future work.
Without loss of generality, we can assume that universal CD-restrictions are not used in concepts, because we can express $\forall p_1,\dots,p_k.P$ as $\lnot\exists p_1,\dots,p_k.P_1\sqcap\dots\sqcap\lnot\exists p_1,\dots,p_k.P_m$, where $P_1,\dots,P_m$ are all $k$-ary predicates of~$\cDom$ except for~$P$ (the union of these predicates is equivalent to the complement of~$P$ due to \JEPD, and there are only finitely many of them since the signature is assumed to be finite).
We can additionally assume that concepts do not contain value restrictions $\forall r.C$ or disjunctions $C\sqcup D$ since they can be expressed using only negation, conjunction and existential restriction.

\section{Consistency in \texorpdfstring{\ALCD}{ALC(D)}}
\label{sec:satisfiability}

Let now $\cDom$ be a fixed \expadm concrete domain, $\Omc=\Amc\cup\Tmc$ be an \ALCD ontology, and \Mmc be the set of all subconcepts appearing in~\Omc and their negations.
For the type elimination algorithm, we start by defining the central notion of \emph{types}, which is standard.
\begin{definition}
  \label{dfn:type}
  A set $t \subseteq \Mmc$ is a \emph{type} w.r.t.~\Omc if it satisfies the following properties:
  \begin{itemize}
    \item if $C\sqsubseteq D\in\Tmc$ and $C \in t$, then $D \in t$;
    \item if $\top\in\Mmc$, then $\top\in t$;
    \item if $\neg D \in \Mmc$, then $D \in t$ iff $\neg D \notin t$;
    \item if $D \sqcap  D' \in \Mmc$, then $D \sqcap D' \in t$ iff $D \in t$ and $D' \in t$.
  \end{itemize}
  Given a model \Imc of \Omc and an individual $d \in \Delta^{\Imc}$, the \emph{type of $d$} w.r.t.~\Omc is the set
  \begin{equation*}
    t_{\Imc}(d) := \big\{ C \in \Mmc \mid d \in C^{\Imc} \big\}.
  \end{equation*}
\end{definition}
Clearly, $t_{\Imc}(d)$ satisfies the four conditions required to be a type w.r.t.~\Omc.
We use this connection between individuals and types to define \emph{augmented types} that represent the relationship between an individual, its role successors, and the CD-restrictions that ought to be satisfied.
Hereafter, let \nex be the number of existential restrictions $\exists r.C$ in~\Mmc, and \ncd the number of CD-restrictions $\exists p_1,\dotsc,p_k. P$ in \Mmc.
The maximal arity of predicates~$P$ occurring in~\Mmc is denoted by~\nar, and we define $\nt := \nex + \ncd \cdot \nar$.
Intuitively, each non-negated existential restriction in a type needs a successor (and associated type) to be realized, while CD-restrictions may require \nar role successors to fulfill a certain constraint.
Therefore, $\nt$ is an upper bound on the number of successors needed to satisfy all the non-negated restrictions occurring in a type~$t$ w.r.t.~\Omc.

Given a type $t_0$, we define a constraint system associated with a sequence of types $t_1,\dotsc,t_{\nt}$ representing the role successors of a domain element with type $t_0$.
This system contains a variable~$f^i$ for each feature~$f$ of an individual with type~$t_i$ in order to express the relevant CD-restrictions.
Concrete features that are not represented in this system can remain undefined since their values are irrelevant for satisfying the CD-restrictions.
\begin{definition}
  \label{dfn:local-system}
  A \emph{local system} for a type $t_0$ w.r.t.\ a sequence of types $t_1,\dotsc,t_{\nt}$
  is a complete constraint system $\cs$ for which there exists a \emph{successor function} $\agsucc\colon \NR(\Omc)\to\mathcal{P}(\{1,\dots,\nt\})$, such that, for all $\exists p_1,\dots,p_k.P\in\Mmc$, the following condition holds:
  \begin{center}
    $\exists p_1,\dotsc,p_k. P \in t_0$ iff there is $P(v_1,\dotsc,v_k) \in \cs$ for some variables $v_1,\dotsc,v_k$ such that
    \begin{equation*}
      v_i =
      \begin{cases}
        f^0 & \text{if $p_i=f$, or}                     \\
        f^j & \text{if $p_i=rf$ and $j\in \agsucc(r)$.}
      \end{cases}
    \end{equation*}
  \end{center}
\end{definition}
We use a \emph{sequence} instead of a \emph{set} of types for the role successors, since there can be TBoxes that require the existence of two successors with the same type that only differ in their feature values.
For example, for the consistent ontology $\Omc:=\{\top\sqsubseteq\exists rf,rf.{<}\}$ over $\stru{Q}=(\mathbb{Q},{<},{=},{>})$, we have $\Mmc=\{\top,\ \neg\top,\ \exists rf,rf.{<},\ \neg\exists rf,rf.{<}\}$, and the only type is $t=\{\top,\ \exists rf,rf.{<}\}$.
Any $r$-successors witnessing $\exists rf,rf.{<}$ for an element in a model of~\Omc have the same type~$t$.
However, we cannot express the restriction on their $f$-values by the (unsatisfiable) constraint ${<}(f^t,f^t)$, but need to consider two copies $t_1,t_2$ of $t$ to get the (satisfiable) constraint ${<}(f^1,f^2)$.

To merge the local systems associated to types of adjacent elements in a model, we introduce the following operation.
For two local systems~$\cs,\cs'$, the \emph{merged system} $\cs \merge{i} \cs'$ is obtained as the union of~$\cs$ and~$\cs'$ where we identify all features with index~$i$ in~$\cs$ with those of index~$0$ in~$\cs'$, while keeping the remaining variables separate.
Formally, we first replace all variables~$f^j$ in~$\cs'$ by fresh variables~$f^{j'}$ and subsequently replace the variables~$f^{0'}$ in~$\cs'$ by~$f^i$.

\begin{definition}
  \label{dfn:augmented-type}
  An \emph{augmented type} for~\Omc is a tuple $\agt := (t_0,\dotsc,t_{\nt},\cs_{\agt})$ where $t_0,\dotsc,t_{\nt}$ are types for~\Omc and $\cs_{\agt}$ is a local system for~$t_0$ w.r.t.\ $t_1,\dotsc,t_{\nt}$. 
  The \emph{root} of~$\agt$ is $\agroot(\agt) := t_0$.
  The augmented type $\agt$ is \emph{locally realizable} if $\cs_{\agt}$ has a solution and if there exists a successor function $\agsucc_{\agt}$
  for $\cs_{\agt}$
  s.t. for all concepts $\exists r. C \in \Mmc$, it holds that
  \begin{center}
    $\exists r. C \in \agroot(\agt)$ iff there is $i \in \agsucc_{\agt}(r)$ such that $C \in t_i$.
  \end{center}
  An augmented type $\agt'$ then \emph{patches} $\agt$ at $i \in \agsucc_{\agt}(r)$ if $\agroot(\agt') = t_i$
  and the system $\cs_{\agt} \merge{i} \cs_{\agt'}$ has a solution.
  A set of augmented types $\mathbb{T}$ \emph{patches} the locally realizable $\agt$ if, for every role name $r$ and every $i \in \agsucc_{\agt}(r)$, there is a $\agt' \in\mathbb{T}$ that patches $\agt$ at $i$.
\end{definition}
For the ontology \Omc introduced above, we have $\nt = 2$ since \Mmc only contains one CD-restriction over a binary predicate.
Using infix notation, all augmented types $\agt = (t,t,t,\cs_\agt)$ for~\Omc are such that $\cs_\agt$ contains the constraints $f^i = f^i$ for $i = 0,1,2$, and either $f^1 < f^2$ or $f^2 < f^1$.
There are augmented types~$\agt$ that are not locally realizable, for instance if $\cs_\agt$ contains $f^0 = f^1$, $f^0 = f^2$, and $f^1 < f^2$.
On the other hand, there is a locally realizable augmented type using the constraints $f^0 < f^1$, $f^1 < f^2$, and $f^0 < f^2$, which can patch itself both at $i \in \{1,2\}$.

To additionally handle named individuals and concept and role assertions, we introduce a structure~$\agt_\Amc$ that describes all ABox individuals and their connections simultaneously, similar to the common notion of \emph{precompletion}.
The associated constraint system~$\cs_\Amc$ now uses variables~$f^{a,i}$ indexed with individual names~$a$ in addition to numbers~$i$.

\begin{definition}\label{dfn:abox-type}
  An \emph{ABox type} for~\Omc is a tuple $\agt_\Amc:=\big((\agt_a)_{a\in \NI(\Amc)},\Amc_R,\cs_\Amc\big)$, where $\agt_a$ are augmented types, $\Amc_R$ is a set of role assertions over $\NI(\Amc)$ and $\NR(\Omc)$, and $\cs_\Amc$ is a complete and satisfiable constraint system such that, for every $a\in \NI(\Amc)$,
  \begin{itemize}
    \item for every concept assertion $C(a)\in\Amc$, we have $C\in\agroot(\agt_a)$;
    \item for every role assertion $r(a,b)\in\Amc$, we have $r(a,b)\in\Amc_R$;
    \item for every $\lnot\exists r.C \in \agroot(\agt_a)$ and $r(a,b)\in\Amc_R$, we have $C\notin\agroot(\agt_b)$;
    \item for every $P(f_1^{j_1},\dots,f_k^{j_k})\in \cs_{\agt_a}$, we have $P(f_1^{a,j_1},\dots,f_k^{a,j_k})\in \cs_\Amc$;
    \item for every $\lnot\exists p_1,\dots,p_k.P\in\agroot(\agt_a)$, there can be no $P(v_1,\dots,v_k)\in \cs_\Amc$ with
          \[ v_i =
            \begin{cases}
              f^{a,0} & \text{if $p_i=f$,}                                \\
              f^{b,0} & \text{if $p_i=rf$ and $r(a,b)\in\Amc_R$, or}      \\
              f^{a,j} & \text{if $p_i=rf$ and $j\in\agsucc_{\agt_a}(r)$;}
            \end{cases}
          \]
  \end{itemize}
\end{definition}
Positive occurrences of existential role or CD-restrictions in the ABox type do not need to be handled, as these are satisfied by anonymous successors described in the augmented types~$\agt_a$.

\subsection{The Type Elimination Algorithm}

\Cref{alg:elimination} uses the introduced notions to check consistency of~\Omc.
\begin{algorithm}[tb]
  \caption{Elimination algorithm for consistency of \ALCD ontologies}
  \label{alg:elimination}
  \begin{algorithmic}[1]
    \Require An \ALCD ontology $\Omc=\Amc\cup\Tmc$
    \Ensure \textsc{consistent} if \Omc is consistent, and \textsc{inconsistent} otherwise
    \State $\Mmc \gets \text{all subconcepts occurring in \Omc and their negations}$
    \State $\mathbb{T} \gets \text{all augmented types for~\Omc}$
    \While{there is $\agt \in \mathbb{T}$ that is not locally realizable or not patched by $\mathbb{T}$}\label{a:loop}
    \State $\mathbb{T} \gets \mathbb{T} \setminus \{\agt\}$
    \EndWhile
    \If{there is an ABox type~$\agt_\Amc$ for~\Omc with $\agt_a \in \mathbb{T}$ for all $a\in \NI(\Amc)$} \label{a:abox}
    \State \textbf{return} \textsc{consistent}
    \Else
    \State \textbf{return} \textsc{inconsistent}
    \EndIf
  \end{algorithmic}
\end{algorithm}

\begin{lemma}[Soundness]
  \label{lem:soundness}
  If \Cref{alg:elimination} returns \textsc{consistent}, then \Omc is consistent.
\end{lemma}
\begin{proof}
  Assume that $\mathbb{T}$ and $\agt_A=((\agt_a)_{a\in \NI(\Amc)},\Amc_R,\cs_\Amc)$ are obtained after a successful run of the elimination algorithm.
  We use them to define a forest-shaped interpretation \Imc that is a model of~\Omc.
  The domain of this model consists of pairs $(a,w)$, where $a\in \NI$ designates a tree-shaped part of~\Imc whose structure is given by the words~$w$ over the alphabet $\Sigma := \mathbb{T} \times \{0,\dotsc,\nt\}$. A pair $(\agt,i)\in\Sigma$ describes an augmented type and the position relative to the restriction that this augmented type fulfills w.r.t.\ its parent in the tree.
  For a word $w \in \Sigma^{+}$, we define $\wend(w) := \agt$ if $(\agt,j)$ occurs at the last position of $w$ for some $j \in \{0,\dotsc,\nt\}$.

  We start defining the domain of~\Imc by $\Delta^0 := \{ (a, w_a) \mid a\in \NI(\Amc),\ w_a:=(\agt_a, 0) \}$. Observe that $w_a\in\Sigma$, since $\agt_a \in \mathbb{T}$.
  Assuming that $\Delta^{m}$ is defined, we define $\Delta^{m+1}$ based on $\Delta^m$, and subsequently construct the domain of \Imc as the union of all sets $\Delta^{m}$.
  Given $(a,w) \in \Delta^m$ with $\wend(w) = \agt$,
  we observe that $\agt$ must have a successor function $\agsucc_{\agt}$ s.t. for every $i \in \agsucc_{\agt}(r)$, there is an augmented type $\mathfrak{u}^i \in \mathbb{T}$ patching $\agt$ at~$i$, as otherwise $\agt$ would have been eliminated from $\mathbb{T}$.
  We use these augmented types to define $\Delta_r^{m+1}[a,w] := \{ (a,w \cdot (\mathfrak{u}^i,i)) \mid i \in \agsucc_{\agt}(r) \}$ to then obtain
  \begin{equation*}
    \Delta^{m+1} := \Delta^m \cup \bigcup \big\{ \Delta_r^{m+1}[a,w] \mid (a,w) \in \Delta^m \;\text{and}\; r \in \NR \big\}
  \end{equation*}
  and set $\Delta^{\Imc} := \bigcup_{m \in \mathbb{N}} \Delta^m$.
  The interpretation of individual, concept, and role names over \Imc is given by
  \begin{align*}
    a^\Imc   & := (a,w_a),                                                                                                                                               \\
    A^{\Imc} & := \big\{ (a,w) \in \Delta^{\Imc} \mid \wend(w) = \agt \;\text{and}\; A \in \agroot(\agt) \big\},                                                         \\
    r^{\Imc} & := \big\{ ((a,w_a),(b,w_b)) \mid r(a,b)\in\Amc_R \big\} \cup{}                                                                                            \\
             & \phantom{{}:={}} \big\{ ((a,w),(a,w')) \mid (a,w) \in \Delta^{m} \;\text{and}\; (a,w') \in \Delta^{m+1}_r[a,w] \; \text{with} \; m \in \mathbb{N} \big\}.
  \end{align*}

  Defining the interpretation of feature names in \Imc requires more work.
  Given $(a,w) \in \Delta^{\Imc}$ with $\wend(w) = \agt$, let $\cs_{a,w}$ be the constraint system obtained by replacing every variable $f^0$ in $\cs_{\agt}$ with $f^{a,w}$ and every other variable $f^i$ in $\cs_{\agt}$ with $f^{a,u}$, where $u \in \Sigma^{+}$ is the unique word of the form $w \cdot (\agt', i)$ for which $(a,u)\in\Delta^{\Imc}$.
  Correspondingly, let $\cs_\Amc^0$ be the result of replacing all variables~$f^{a,0}$ in~$\cs_\Amc$ by $f^{a,w_a}$ and $f^{a,i}$ by $f^{a,u}$, where $u$ is the unique word of the form $w_a\cdot(\agt',i)$ for which $(a,u)\in\Delta^\Imc$.
  For $m \in \mathbb{N}$, let $\cs^m$ be the union of $\cs_\Amc^0$ and all constraint systems $\cs_{a,w}$ for which $(a,w) \in \Delta^m$.
  The proofs of the following claims can be found \inAppendix.
  \begin{restatable}{claim}{FirstClaim}
    \label{claim:patchwork}
    For every $m \in \mathbb{N}$, the constraint system $\cs^m$ has a solution.
  \end{restatable}

  Using this claim, we show how to define an interpretation of feature names for \Imc.
  Let $\cs^{\Imc}$ be the union of all systems $\cs^m$ for $m \in \mathbb{N}$.
  Every finite system $\ocs \subseteq \cs^{\Imc}$ is also a subsystem of $\cs^m$ for some $m \in \mathbb{N}$.
  Since $\cs^m$ has a solution, it follows that $\ocs$ has a solution.
  Every finite subsystem of $\cs^{\Imc}$ has a solution; since $\cDom$ has the homomorphism $\omega$-compactness property, we infer that $\cs^{\Imc}$ has a solution $h^{\Imc}$.
  Using this solution, we define for every feature name $f$ the interpretation $f^{\Imc}(a,w) := h^{\Imc}(f^{a,w})$ if $f^{a,w}$ occurs in $\cs^{\Imc}$, and leave it undefined otherwise.

  \begin{restatable}{claim}{SecondClaim}
    \label{claim:soundness}
    If $C \in \Mmc$ and $(a,w) \in \Delta^{\Imc}$ with $\wend(w) = \agt$, then
    $C \in \agroot(\agt)$ iff $(a,w) \in C^{\Imc}$.
  \end{restatable}
  By this claim and Definition~\ref{dfn:type}, we know that \Imc is a model of \Tmc.
  By~\Cref{dfn:abox-type} we obtain that for each $C(a)\in\Amc$, we have $C\in\agroot(\agt_a)$, and thus $a^\Imc=(a,w_a)=(a,(\agt_a,0))\in C^\Imc$.
  Similarly, whenever $r(a,b)\in\Amc$, then $r(a,b)\in\Amc_R$, and thus $(a^\Imc,b^\Imc)=((a,w_a),(b,w_b))\in r^\Imc$ by the construction of~\Imc.
  Therefore, we conclude that \Imc is also a model of \Amc, and thus of \Omc.
\end{proof}
\begin{lemma}[Completeness]
  \label{lem:completeness}
  If \Omc is consistent, then~\Cref{alg:elimination} returns \textsc{consistent}.
\end{lemma}
\begin{proof}
  Let \Imc be a model of \Omc and define the set $T_{\Imc}$ of all types that are realized in \Imc, that is,
  \begin{equation*}
    T_{\Imc} := \big\{ t_{\Imc}(d) \mid d \in \Delta^{\Imc} \big\}.
  \end{equation*}
  Given that \Imc is a model of \Omc, every element of $T_{\Imc}$ is a type according to~\Cref{dfn:type}.
  Using the elements of $T_{\Imc}$, we construct a set $\mathbb{T}_{\Imc}:=\{\agt_{\Imc}(d)\mid d\in\Delta^{\Imc}\}$ of augmented types.
  For any domain element $d\in \Delta^{\Imc}$, we build the augmented type $\agt_{\Imc}(d)=(t_0,\dotsc,t_{\nt},\cs_d)$ as follows.

  First, we set $t_0 := t_{\Imc}(d) \in T_{\Imc}$.
  Assuming that $\exists r_i. C_i \in \Mmc$ for $i = 1,\dotsc,\nex$, we select types $t_1,\dotsc,t_{\nex}$ to add to $\agt$ that realize the (possibly negated) existential role restrictions occurring in $t_0$.
  If $\exists r_i. C_i \in t_0$, then we can select $t_i$ as the type of an $r$-successor $d'$ of $d$ such that $d' \in C_i^{\Imc}$; otherwise, we pick $t_i$ as the type of an arbitrary element in $\Delta^{\Imc}$.

  Next, we assume that $\exists p_1^i,\dotsc,p_k^i. P_i \in \Mmc$ for $i = 1,\dotsc,\ncd$ and define the function $\mathsf{off}(i,j) := \nex + (i - 1) \cdot \nar + j$ to be able to refer to the $j$-th path in the $i$-th CD-restriction above.
  We select types $t_{\mathsf{off}(i,1)},\dotsc,t_{\mathsf{off}(i,\nar)}$ that realize the (possibly negated) existential CD-restrictions occurring in $t_0$ for $i = 1,\dotsc,\ncd$.
  If $\exists p^i_1,\dotsc,p^i_k. P_i \in t_0$, then there exist values $v^i_j \in {(p^i_j)}^{\Imc}(d)$ for $j = 1,\dotsc,k$ such that $(v^i_1,\dotsc,v^i_k) \in P^D$.
  If $p_j = r f$ holds for some feature name $f$ and some role name $r$, let $t_{\mathsf{off}(i,j)}$ be the type of an $r$-successor $d'$ of $d$ such that $v_j^i = f^{\Imc}(d)$.
  For every $j = 1,\dotsc,\nar$ for which $t_{\mathsf{off}(i,j)}$ has not been selected this way, let $t_{\mathsf{off}(i,j)}$ be the type of an arbitrary individual in $\Delta^\Imc$.
  Similarly, if $\exists p^i_1,\dotsc,p^i_k. P_i \notin t_0$ then we set $t_{\mathsf{off}(i,j)}$ be the type of an arbitrary individual in $\Delta^\Imc$ for $j = 1,\dotsc,\nar$.

  The two processes described in the two previous paragraphs yield a sequence of types $t_1,\dotsc,t_{\nt}$ that occur in $T_{\Imc}$ and can thus be associated to individuals $d_1,\dotsc,d_{\nt} \in \Delta^{\Imc}$.
  Using these individuals, we define the local system associated to our augmented type~$\agt_{\Imc}(d)$.
  First, we define the constraint system $\cs_d$ that contains the constraint $P\big(f_1^{i_1},\dotsc,f_k^{i_k}\big)$ iff $\big(f_1^{\Imc}(d_{i_1}),\dotsc, f_k^{\Imc}(d_{i_k})\big) \in P^D$ for all $i_1,\dotsc,i_k \in \{0,\dotsc,\nt\}$.
  We associate to this constraint system a successor function $\agsucc_d$ that assigns to $r \in \NR$ all $i \in\{ 1,\dotsc,\nt\}$ for which $d_i$ is an $r$-successor of~$d$.
  This concludes our definition of $\agt_{\Imc}(d)$ for $d \in \Delta^{\Imc}$, and thus of $\mathbb{T}_{\Imc}$.
  \begin{restatable}{claim}{ThirdClaim}
    Every augmented type in $\mathbb{T}_{\Imc}$ is locally realizable and patched in $\mathbb{T}_{\Imc}$.
  \end{restatable}
  Therefore, no augmented type $\agt\in\mathbb{T}_{\Imc}$ is eliminated during a run of \Cref{alg:elimination}, and so $\mathbb{T}_\Imc \subseteq \mathbb{T}$.
  We further deduce that $\mathbb{T}$ cannot become empty, since $T_{\Imc}$ is non-empty.
  Using $\mathbb{T}_{\Imc}$ together with our model \Imc of $\Omc$ we derive an ABox type $\agt_{\Amc}^{\Imc}=\big((\agt_a)_{a\in \NI(\Amc)},\Amc_R,\cs_\Amc\big)$ for $\Omc$.
  We define each $\agt_a$, $a \in \NI(\Amc)$, as $\agt_a := \agt_{\Imc}(a^{\Imc}) \in \mathbb{T}_{\Imc}$.
  Assuming that $d_{a,i}$ is the domain element used to establish the $i$-th type in $\agt_a$ for $i \in \{0,\dotsc,\nt\}$, we define the constraint system $\cs_\Amc$ s.t.
  \begin{equation*}
    P(f_1^{a_1,i_1},\dotsc,f_k^{a_k,i_k}) \in \cs_\Amc \;\text{iff}\; (f_1^{\Imc}(d_{a_1,i_1}),\dotsc,f_k^{\Imc}(d_{a_k,i_k})) \in P^D.
  \end{equation*}
  Finally, we define the set $\Amc_R$ to consist of all role assertions $r(a,b)$ for which $a, b \in \NI(\Amc)$, $r \in \NR(\Omc)$, and $(a^{\Imc},b^{\Imc}) \in r^{\Imc}$.
  \begin{restatable}{claim}{FourthClaim}
    The object $\agt_{\Amc}^{\Imc}$ is an ABox type for \Omc with $\agt_a \in \mathbb{T}_{\Imc}$ for all $a \in \NI(\Amc)$.
  \end{restatable}
  Thus, there is a suitable ABox type for \Omc, and \Cref{alg:elimination} returns \textsc{consistent}.
\end{proof}
We can show in a standard way that~\Cref{alg:elimination} runs in exponential time w.r.t. \Omc and obtain the following theorem (a detailed proof can be found \inAppendix).

\begin{restatable}{theorem}{ThmConsistencyExpTime}
  \label{thm:consistency-exptime}
  Let $\cDom$ be an \expadm concrete domain.
  Then, the consistency problem for \ALCD ontologies is \ExpTime-complete.
\end{restatable}

\section{Concrete Domain Assertions}
\label{sec:assertions}
Beside using concrete domain restrictions on the concept level, we may want to use feature and predicate assertions in the ABox to either assign a specific value of a feature to some individual or directly constraint the values of features of different individuals. Formally, a \emph{feature assertion} is of the form $f(a,c)$, where $f\in\NF$, $a\in\NI$, and $c\in D$, and it is satisfied by an interpretation \Imc if $f^\Imc(a^\Imc)=c$. \emph{Predicate assertions} are of the form $P\big(f_1(a_1),\dotsc,f_k(a_k)\big)$, where $P$ is a $k$-ary relation over $\cDom$ and $f_i \in \NF$, $a_i \in \NI$ for $i = 1,\dotsc,k$. An interpretation~\Imc satisfies such an assertion if $\big(f_1^{\Imc}(a_1^{\Imc}),\dotsc,f_k^{\Imc}(a_k^{\Imc})\big) \in P^D$.
In our setting, we can simulate predicate assertions using concept and role assertions, which leads to the following result.

\begin{theorem}\label{thm:cd-assertions}
  For any \expadm concrete domain~$\cDom$, ontology consistency in \ALCD with predicate assertions is \ExpTime-complete.
\end{theorem}
\begin{proof}
  First, we show how to simulate CD-restrictions of the form $\exists f.\top_\cDom$ that describe all individuals~$d$ in \Imc for which $f^\Imc(d)$ is defined.
  Although $\top_\cDom$ may not be a predicate of $\cDom$, by \JEPD and the fact that the signature of~$\cDom$ is non-empty and finite, $\top_\cDom$ can be expressed as the disjunction of some $k$-ary predicates $P_1,\dots,P_m$.
  This implies that for every $d \in D$ there is exactly one $k$-ary predicate $P_i$ such that $(d,\dotsc,d) \in P_i^D$.
  Thus, we can write $\exists f.\top_\cDom$ equivalently as ${\exists f,\dotsc,f.P_1}\sqcup{\dotsb}\sqcup{\exists f,\dotsc,f.P_m}$, where each restriction $\exists f,\dots,f.P_i$ repeats $f$ for $k$ times.

  Let now $\Omc=\Amc\cup\Tmc$ be an \ALCD ontology with predicate assertions.
  We introduce a fresh individual name~$a^*$ and fresh role names~$r_a$ for all individual names $a\in\NI(\Omc)$.
  The ontology $\Omc'$ is then obtained from \Omc by adding the role assertions $r_a(a^*,a)$ for all $a\in\NI(\Omc)$, and replacing all predicate assertions $P\big(f_1(a_1),\dotsc,f_k(a_k)\big)$ in~\Amc by the concept assertions $(\exists f_1.\top_\cDom)(a_1)$, \dots, $(\exists f_k.\top_\cDom)(a_k)$, and $(\forall r_{a_1}f_1,\dots,r_{a_k}f_k.P)(a^*)$.
  Any model~$\Imc$ of~$\Omc'$ satisfies $(c_1,\dots,c_k)\in P^D$ for all possible values $c_i\in (r_{a_i}f_i)^\Imc$, and thus in particular for $c_i=f_i^\Imc(a_i^\Imc)$, which shows that \Imc is also a model of~\Omc.
  Conversely, from every model~\Imc of~\Omc we obtain a model of~$\Omc'$ by choosing an arbitrary element $d^*\in\Delta^\Imc$ for the interpretation of~$a^*$ and adding $(d^*,a^\Imc)$ to the interpretation of~$r_a$ for every individual name~$a$.
\end{proof}

Feature assertions can also be simulated under certain conditions.
For concrete domains that contain \emph{singleton} predicates $=_c$ with $(=_c)^D=\{c\} \subseteq D$, we can express any feature assertion $f(a,c)$ using the concept assertion $(\exists f.{=}_c)(a)$.
However, due to its finite signature, $\cDom$ can only contain finitely many such predicates, and hence the feature assertions are restricted by the chosen concrete domain.
Due to the \JD and \JEPD conditions, it turns out that adding feature assertions is actually \emph{equivalent} to adding singleton predicates in the following sense.
Here, an \emph{additional singleton predicate}~$=_c$ is one that is not part of~$\cDom$, but otherwise can be used in an ontology with the same semantics as defined above; the proof can be found \inAppendix.

\begin{restatable}{theorem}{ThmFeatureAssertions}\label{theo:feature-assertions}
  For an $\omega$-admissible concrete domain~$\cDom$, ontology consistency in \ALCD with additional singleton predicates can be polynomially reduced to ontology consistency in \ALCD with feature assertions.
\end{restatable}

If we additionally require~$\cDom$ to be homogeneous, then we can show that arbitrary feature assertions can already be expressed in ordinary \ALCD ontologies.

\begin{restatable}{theorem}{ThmHomogeneous}\label{thm:homogeneous}
  For an \expadm homogeneous concrete domain~$\cDom$, ontology consistency in \ALCD with feature assertions is \ExpTime-complete.
\end{restatable}
\begin{proof}
  Due to Theorem~\ref{thm:cd-assertions}, it suffices to provide a reduction to ontology consistency with predicate assertions.
  Let $\Omc = \Amc \cup \Tmc$ be an \ALCD ontology containing feature assertions.
  Let $\Amc'$ be the ABox containing all concept and role assertions from~$\Amc$ and, in addition, all predicate assertions $P(f_1(a_1),\dotsc,f_k(a_k)) \in \Amc'$ for all combinations of feature assertions $f_i(a_i,c_i) \in \Amc$ with $i = 1,\dotsc,k$ for which $(c_1,\dotsc,c_k) \in P^D$ holds.
  The size of $\Amc'$ is polynomial in the input, since the signature of~$\cDom$ is fixed.

  It is easy to see that every model of $\Omc$ is also a model of $\Omc' := \Amc' \cup \Tmc$.
  Conversely, let \Imc be a model of $\Omc'$ and let $\cDom_\Amc$, $\cDom_\Imc$ be the finite substructures of $\cDom$ over the domains
  \begin{center}
    $D_\Amc := \{ c \mid f(a,c) \in \Amc \}$ and $D_\Imc := \{ f^{\Imc}(a) \mid f(a,c) \in \Amc \}$,
  \end{center}
  respectively.
  By definition of~$\Amc'$ and \JEPD, we know that $\big(f_1^{\Imc}(a_1^\Imc),\dotsc,f_k^{\Imc}(a_k^\Imc)\big) \in P^D$ iff $(c_1,\dotsc,c_k) \in P^D$, for all combinations of feature assertions $f_i(a_i,c_i)$ in~\Amc.
  By~\JD, this in particular implies that $f_1^\Imc(a_1^\Imc)=f_2^\Imc(a_2^\Imc)$ iff $f_1(a_1,c),f_2(a_2,c)\in \Amc$ for some value $c\in D$, which means that the two substructures have the same number of elements.
  Moreover, by the first equivalence, the mapping $f^{\Imc}(a^\Imc) \mapsto c$ for all $f(a,c) \in \Amc$ is an isomorphism between $\cDom_\Imc$ and $\cDom_\Amc$.
  Since $\cDom$ is homogeneous, there exists an isomorphism $h \colon D \to D$ such that $h(f^{\Imc}(a^\Imc)) = c$ if $f(a,c) \in \Amc$.
  Consequently, we define $\Imc'$ from \Imc by changing the interpretation of feature names to $f^{\Imc'}(d) := h(f^{\Imc}(d))$ iff this value is defined for $f \in \NF$ and $d \in \Delta^{\Imc}$.
  Since $h$ is an isomorphism, we have $C^\Imc=C^{\Imc'}$ for all concepts~$C$, including CD restrictions, which shows that $\Imc'$ is a model of~\Tmc and all concept and role assertions in~\Amc.
  Moreover, it also satisfies all feature assertions $f(a,c)\in\Amc$ since $f^{\Imc'}(a^{\Imc'})=h(f^\Imc(a^\Imc))=c$ by construction.
\end{proof}

Together with~\Cref{theo:feature-assertions}, this also shows that one can use arbitrary singleton equality predicates in \ALCD ontologies over a homogeneous concrete domain~$\cDom$, even if~$\cDom$ contains only finitely many singleton predicates (or none).

\section{Conclusion}
\label{sec:conclusion}

In this paper, we revisited the problem of reasoning in \ALCD with an $\omega$-admissible concrete domain~$\stru{D}$, first addressed in~\cite{LuMi07}.
There, it was conjectured that concept satisfiability w.r.t.\ a TBox is \ExpTime-complete, provided that $\CSP(\stru{D})$ is decidable in exponential time.
Using an approach based on type elimination, we successfully proved this conjecture.
In addition, we integrated ABox reasoning and showed that reasoning w.r.t.\ an \ALCD ontology where one can refer to specific values via feature assertions is also \ExpTime-complete, if in addition to the above $\stru{D}$ is an $\omega$-admissible homogeneous structure.
The main examples of $\omega$-admissible concrete domains from the literature fulfill this requirement as they are (reducts of) finitely bounded homogeneous structures~\cite{BaRy22}, and so we obtain insights into the complexity of reasoning with extensions of \ALC by concrete domain restrictions ranging over Allen's interval algebra, the region connection calculus RCC8, the rational numbers with ordering and equality, and disjoint combinations of those domains.

By extending the type elimination algorithm proposed in this paper appropriately, we believe that it is possible to show that decidability is preserved in extensions of \ALCD such as $\mathcal{ALCI}(\stru{D})$, where inverse roles are allowed in both role and concrete domain restrictions (so that we can write, e.g.\ $\forall \ex{hasChild}^{-}\,\ex{age},\ex{hasChild}\,\ex{age}.{>}$), and $\mathcal{ALC\!Q}(\stru{D})$, which supports qualified number restrictions.
We also plan to use this type elimination approach as a starting point in our investigations of different inference problems, for instance signature-based abduction~\cite{DBLP:conf/ijcai/Koopmann21} for \ALCD ontologies and abstract definability for \ALCD TBoxes, i.e. checking whether their \emph{abstract expressive power}~\cite{BaBo-SAC-24} can be defined in \ALC.

\begin{acknowledgments}
  This work was partially supported by DFG grant 389792660 as part of \href{https://perspicuous-computing.science}{TRR~248 -- CPEC}, by the German Federal Ministry of Education and Research (BMBF, SCADS22B) and the Saxon State Ministry for Science, Culture and Tourism (SMWK) by funding the competence center for Big Data and AI "\href{https://scads.ai}{ScaDS.AI Dresden/Leipzig}".
  The authors would like to thank Jakub Rydval for his help in understanding the properties of finitely bounded homogeneous structures.
\end{acknowledgments}

\clearpage
\appendix

\section{Omitted Proofs for \Cref{sec:satisfiability}}
\label{sec:appendix}

\FirstClaim*
\begin{proof}
  We prove the claim by induction over $m \in \mathbb{N}$.
  For the base case $m = 0$, observe that $\cs^0$ is equal to $\cs_\Amc^0$ since $\cs_\Amc^0$ already contains all local systems of the form $\cs_{a,w_a}$ (see Definition~\ref{dfn:abox-type}).
  Since $\cs_\Amc^0$ is equal to~$\cs_\Amc$ up to renaming of variables, the fact that $\cs_\Amc$ has a solution implies that $\cs^0=\cs_\Amc^0$ has a solution as well.

  For the inductive step, we assume that $\cs^m$ has a solution and show how to extend it to a solution of $\cs^{m+1}$.
  We begin by observing that any constraint system~$\cs$ that has a solution~$h$ can be extended to a complete constraint system by using~$h$ to add any missing constraints; i.e.\ if there is no constraint $P(v_1,\dots,v_k)$ for $v_1,\dots,v_k\in V(\cs)$, but $\cDom$ has $k$-ary predicates, then we can complete $\cs$ by adding the unique $P(v_1,\dots,v_k)$ for which $(h(v_1),\dots,h(v_k))\in P^D$ (cf.\ \JEPD).
  Moreover, this complete constraint system also has~$h$ as a solution.
  Since $\cs^m$ has a solution, let now $\ocs$ be the satisfiable, complete system obtained by extending $\cs^m$ in this way.
  
  Let $(a,w) \in \Delta^{m+1} \setminus \Delta^m$.
  By construction, there is a unique non-empty word $w' \in \Delta^m$ and a symbol $(\agt,i) \in \Sigma$ such that $w = w' \cdot (\agt,i)$.
  We notice that $\ocs$ and $\cs_{a,w}$ are complete systems that agree on the constraints over their shared variables $V(\ocs)\cap V(\cs_{a,w})$.
  This holds since all constraints over the shared variables (which are all of the form $f^{a,w}$) must occur in~$\ocs$ inside~$\cs_{a,w'}$; for the case of $w'=w_a$, this is because all constraints from $\cs_{a,w_a}$ must occur in $\cs_\Amc^0$ by Definition~\ref{dfn:abox-type}, and $\cs_{a,w_a}$ is complete.
  Both $\wend(w')$ and $\agt$ belong to $\mathbb{T}$, and $\agt$ patches $\wend(w')$ at $i$, thus $\cs_{a,w'} \cup \cs_{a,w}$ has a solution (cf.\ Definition~\ref{dfn:augmented-type}).
  In particular, the relations over the concrete domain $\cDom$ satisfy \JEPD, hence there cannot be a tuple of variables $v_1,\dotsc,v_k$ such that $P(v_1,\dotsc,v_k) \in \cs_{a,w}$ and $P'(v_1,\dotsc,v_k) \in \cs_{a,w'}\subseteq \ocs$ with $P \ne P'$.
  Since $\ocs$ and $\cs_{a,w}$ are complete, agree on the constraints over their shared variables, and both have a solution ($\ocs$ by inductive hypothesis, and $\cs_{a,w}$ because $\agt \in \mathbb{T}$), property \AP implies that $\ocs\cup \cs_{a,w}$ has a solution, which we can use to extend $\ocs$ to a complete constraint that includes~$\cs_{a,w}$.
  
  We can repeat this process for every $(a,w) \in \Delta^{m+1} \setminus \Delta^m$, because the different constraint systems $\cs_{a,w}$ do not share variables, and thus we obtain a constraint system $\ocs'$ that is complete, has a solution, and includes $\cs^{m+1}$.
  Therefore, we conclude that $\cs^{m+1}$ has a solution.
\end{proof}
\SecondClaim*
\begin{proof}
  We prove this claim by structural induction over $C \in \Mmc$.
  We first prove the two base cases where $C$ is either a concept name or an existential CD-restriction.
  \begin{itemize}
    \item The case $C = A$ is trivially covered by the definition of $A^{\Imc}$.
    \item Let $C = \exists p_1,\dotsc,p_k. P \in \Mmc$.
    If $C \in \agroot(\agt)$, by definition of local system and because of $\agt \in \mathbb{T}$ we are able to find a constraint $P(f_1^{j_1},\dotsc,f_k^{j_k}) \in \cs_{\agt}$ such that for each $i = 1,\dotsc,k$
    \begin{itemize}
      \item if $p_i = f_i$, then $j_i = 0$;
      \item if $p_i = r_i f_i$, then $j_i \in \agsucc_{\agt}(r_i)$ and there exists $\mathfrak{u}^i \in \mathbb{T}$ that patches $\agt$ at $j_i$.
    \end{itemize}
    Using these indices and augmented types, we define the words
    \begin{equation*} w^i :=
      \begin{cases}
        w & \text{if $p_i = f_i$}, \\
        w \cdot (\mathfrak{u}^i, j_i) & \text{if $p_i = r_i f_i$}.
      \end{cases}
    \end{equation*}
    for $i = 1,\dotsc,k$.
    It follows that $P(f_1^{a,w^1},\dotsc,f_k^{a,w^k}) \in \cs^{\Imc}$, which implies that $\big(f_1^{\Imc}(a,w^1),\dotsc,f_k^{\Imc}(a,w^k)\big) \in P^D$ by definition of \Imc.
    Due to the construction of~$w^i$, we know that $f_i^{\Imc}(a,w^i) \in p_i^{\Imc}(a,w)$ holds for $i = 1,\dotsc,k$, which allows us to conclude that $(a,w) \in C^{\Imc}$.
    
    Vice versa, assume that $C \notin \agroot(\agt)$.
    By Definition~\ref{dfn:type}, this means that $\lnot C \in \agroot(\agt)$.
    For every tuple of values $c_i \in p_i^{\Imc}(w)$, $i = 1,\dotsc,k$, we have to show that $(c_1,\dots,c_k)\notin P^D$.
    
    We first consider the case that $w=w_a$ for some $a\in \NI(\Amc)$, and find individual names~$a^i$ and words~$w^i$ such that the variable $f^{a^i,w^i}$ reflects the origin of the value~$c_i$, as follows.
    By construction of~\Imc, one of the following three cases must hold for each $i\in\{1,\dots,k\}$.
    \begin{itemize}
      \item If $p_i=f_i$ and $c_i=f_i^\Imc(a,w_a)=h^\Imc(f_i^{a,w_a})$, then we set $a^i:=a$ and $w^i:=w_a$.
      \item If $p_i=r_if_i$ and $c_i=f_i^\Imc(b,w_b)=h^\Imc(f_i^{b,w_b})$ for some individual name $b$ with $r_i(a,b)\in\Amc_R$, then we set $a^i:=b$ and $w^i:=w_b$.
      \item If $p_i=r_if_i$ and $c_i=f_i^\Imc(a,w')=h^\Imc(f_i^{a,w'})$ for some $r_i$-successor $(a,w')$ of $(a,w)$ with $w'=w\cdot(\agt',j)$ and $j\in\agsucc_{\agt}(r_i)$, then we set $a^i:=a$ and $w^i:=w'$.
    \end{itemize}
    By Definition~\ref{dfn:abox-type} and the fact that $\lnot C \in \agroot(\agt)$, we know that $P\big(f_1^{a^1,w^1},\dots,f_k^{a^k,w^k}\big)$ cannot be contained in~$\cs_\Amc^0$.
    However, by construction of~\Imc, the fact that the corresponding feature values are defined means that the variables $f_1^{a^1,w^1},\dots,f_k^{a^k,w^k}$ must all occur in~$\cs_\Amc^0$.
    Hence, by completeness of $\cs_\Amc^0$, then, $P'\big(f_1^{a^1,w^1},\dots,f_k^{a^k,w^k}\big) \in \cs_\Amc^0\subseteq \cs^\Imc$ for some $k$-ary predicate~$P'$ disjoint with~$P$.
    This means that $(c_1,\dotsc,c_k) = \big(h^\Imc(f_1^{a^1,w^1}),\dots,h^\Imc(f_k^{a^k,w^k})\big)\in (P')^D$, and thus $(c_1,\dotsc,c_k) \notin P^D$.
    This allows us to conclude that $(a,w_a) \notin C^{\Imc}$.
    
    For the case that $w\neq w_a$ for all $a\in \NI(\Amc)$, we can use similar, but simpler, arguments, based on~$\cs_w$ (Definition~\ref{dfn:local-system}) instead of~$\cs_\Amc^0$ (Definition~\ref{dfn:abox-type}).
  \end{itemize}
  Let us now assume that the claim holds for $D,C_1,C_2 \in \Mmc$ and prove the inductive cases.
  \begin{itemize}
    \item If $C = \neg D \in \Mmc$, then $C \in \agroot(\agt)$ iff $D \notin \agroot(\agt)$ iff $(a,w) \notin D^{\Imc}$ iff $(a,w) \in C^{\Imc}$, where the equivalences hold due to Definition~\ref{dfn:type}, the inductive hypothesis, and the semantics of negation, respectively.
    \item If $C = C_1 \sqcap C_2 \in \Mmc$, then $C \in \agroot(\agt)$ iff $C_1 \in \agroot(\agt)$ and $C_2 \in \agroot(\agt)$ iff $(a,w) \in C_1^{\Imc}$ and $(a,w) \in C_2^{\Imc}$ iff $(a,w) \in C^{\Imc}$, where the equivalences hold similarly to the case of concept negation.
    \item If $C = \exists r. D \in \Mmc$, then $C \in \agroot(\agt)$ and $\agt \in \mathbb{T}$ imply that there are $i \in \agsucc_{\agt}(r)$ and an augmented type $\agt'\in\mathbb{T}$ that patches $\agt$ at $i$ such that $((a,w),(a,w')) \in r^{\Imc}$ with $w' = w \cdot (\agt', i)$ and $D \in \agroot(\agt')$.
    By definition of \Imc and inductive hypothesis, we deduce that $(a,w') \in D^{\Imc}$, which in turn implies $(a,w) \in C^{\Imc}$.

    Vice versa, assume that $C \notin \agroot(\agt)$, and thus $\neg C \in \agroot(\agt)$.
    Any $r$-successor~$e$ of $(a,w)$ must be of one of the following forms.
    \begin{itemize}
      \item If $w=w_a=(\agt_a,0)$, $e=(b,w_b)$, and $r(a,b)\in\Amc_R$, then $\lnot C \in \agroot(\agt_a)$ implies that $D \notin \agroot(\agt_b)$ by Definition~\ref{dfn:abox-type}. Since $\wend(w_b)=\agt_b$, we thus obtain $e \notin D^\Imc$ by inductive hypothesis.
      \item If $e=(a,w')$ and $\wend(w') = \agt'$ patches some $i \in \agsucc_{\agt}(r)$, then $\neg C \in \agroot(\agt)$ implies that $D \notin \agroot(\agt')$ (see Definition~\ref{dfn:augmented-type}). By inductive hypothesis, we conclude that $e \notin D^\Imc$.
    \end{itemize}
    This shows that no $r$-successor of $(a,w)$ can satisfy~$D$ in~\Imc, and thus $(a,w) \notin C^\Imc$. \qedhere
  \end{itemize}
\end{proof}

\ThirdClaim*
\begin{proof}
  We notice that $\cs_d$ is complete and that $\cs_d$ and $\agsucc_d$ satisfy all the conditions stated in~\Cref{dfn:local-system} and thus constitute a local system.
  In addition, the mapping $v_d(f^i):=f^{\Imc}(d_i)$ is a solution of~$\cs_d$, and thanks to our choice of types $t_0,\dotsc,t_{\nt}$ we also obtain that the augmented type $\agt_{\Imc}(d)$ constructed using this process is locally realizable.

  Moreover, for a given augmented type $\agt=\agt_{\Imc}(d)$ in $\mathbb{T}_{\Imc}$ and $i \in \{0,\dotsc,\nt\}$, our construction yields that $\agt'=\agt_{\Imc}(d_i)$ patches $\agt$ at $i$; here, $d_i$ is the domain element chosen for $t_i$ in the construction of $\agt_{\Imc}(d)$.
  The first condition, namely that $\agroot(\agt') = t_i$, is fulfilled by construction.
  The second condition, i.e.\ that $\cs_{\agt} \merge{i} \cs_{\agt'}$ has a solution, follows from the fact that the individual solutions $v_d$, $v_{d_i}$ constructed above agree on the values of the shared variables $v_d(f^i)=f^{\Imc}(d_i)=v_{d_i}({f^0})$.
  Together with the conditions proved during the construction of the augmented types, this ensures that $\agt$ is patched by $\mathbb{T}_{\Imc}$.
\end{proof}

\FourthClaim*
\begin{proof}
  We show that $\agt^\Imc_\Amc$ satisfies all conditions stated in~\Cref{dfn:abox-type}.
  If $C(a) \in \Amc$ then $a^{\Imc} \in C^{\Imc}$ and since $\agroot(\agt_a) = t_{\Imc}(a^{\Imc})$, we deduce that $C \in \agroot(\agt_a)$.
  Similarly, if $r(a,b) \in \Amc$ then $(a^{\Imc},b^{\Imc}) \in r^{\Imc}$ holds, which by definition of $\Amc_R$ implies that $r(a,b) \in \Amc_R$.
  If $\neg \exists r. C \in \agroot(\agt_a)$ and $r(a,b) \in \Amc_R$, the fact that $b$ is an $r$-successor of $a$ in \Imc clearly implies that $C \notin \agroot(\agt_b) = t_{\Imc}(b^{\Imc})$ must hold.
  We turn our attention to $\cs_{\Amc}$.
  If $P(f_1^{j_1},\dotsc,f_k^{j_k}) \in \cs_{\agt_a}$, then $(f_1^{\Imc}(d_{a,j_1}),\dotsc,f_k^{\Imc}(d_{a,j_k})) \in P^D$, and by definition of~$\cs_{\Amc}$ we obtain $P(f_1^{a_1,i_1},\dotsc,f_k^{a_k,i_k}) \in \cs_\Amc$.
  Next, assume that $a \in \NI(\Amc)$ and that $P(v_1,\dotsc,v_k) \in \cs_{\Amc}$ holds, where each variable $v_i$ is of one of the forms described in~\Cref{dfn:abox-type} for $a$.
  Then, by definition of $\cs_{\Amc}$ we can find values $c_i \in p_i^{\Imc}(a^{\Imc})$ for $i = 1,\dotsc,k$ such that $(c_1,\dotsc,c_k) \in P^D$.
  Therefore, it follows that $\exists p_1,\dotsc,p_k. P \in \agroot(\agt_a)$.
  At last, we observe that $\cs_\Amc$ has a solution, given by the interpretation of feature names over \Imc.
  We conclude that $\agt^\Imc_\Amc$ is an ABox type for \Omc.
\end{proof}

\ThmConsistencyExpTime*
\begin{proof}
	It only remains to show that \Cref{alg:elimination} runs in exponential time. 
	Each augmented type $\agt$ contains $n_\Omc+1$ types, each of size polynomial in the input ontology. Moreover, fixing the types in $\agt$, there are at most $n_{\cDom} \cdot \lvert V_\ell\rvert^{\nar}$ distinct local systems, where $n_\cDom$ is the number of predicates in $\cDom$ and $V_\ell = \nt \cdot \lvert\NF\rvert$ the maximal number of variables. It follows that the number of augmented types is at most exponential, and thus the loop in \Cref{a:loop} takes at most exponentially many iterations. In each iteration, we need to test whether some $\agt\in\mathbb{T}$ is not locally realizable or not patched, which amounts to a number of tests polynomial in $\lvert\mathbb{T}\lvert$, and thus exponential in the size of $\Omc$. Each such test involves checking the satisfiability of complete constraint systems $\cs_{\agt}$ and merged systems $\cs_{\agt} \merge{i} \cs_{\agt'}$. Since $\cDom$ is fixed, each (merged) system is of polynomial size. Since CSP(\cDom) is in \ExpTime, each satisfiability test therefore takes at most exponential time. It follows that the loop in Line~\ref{a:loop} takes at most exponential time in total. It remains to show that also the search for ABox types in Line~\ref{a:abox} takes at most exponential time. For this, it suffices to observe that there can be at most exponentially many ABox types for $\Omc$, which each have a polynomial-size constraint system attached. This is not hard to see from \Cref{dfn:abox-type}, since there is one augmented type for every individual in $\Amc$, the set of possible role assertions in $\Amc_R$ is polynomial in $\Omc$, and the set of variables in $\mathfrak{C}_\Amc$ is polynomial in the size of the involved augmented types. 
\end{proof}

\section{Omitted Proofs for \Cref{sec:assertions}}

\ThmFeatureAssertions*
\begin{proof}
	Let $\Omc$ be an $\ALC(\cDom)$ ontology. We show how to replace all occurrences of an additional singleton predicate~$=_c$ using feature assertions. We first replace every CD-restriction $\exists rf.{=_c}$ by the equivalent concept $\exists r.\exists f.{=_c}$, which means that $=_c$ can occur only in CD-restrictions of the form $\exists f.{=_c}$ with feature paths of length~$1$. The idea is now to store the value~$c$ in a special feature~$f_c$ by using feature assertions in the ABox, and make sure that the value of~$f_c$ is equal to~$c$ at every element reachable from an ABox individual by a role chain. We can then express the concept $\exists f.{=_c}$ by making $f$ equal to $f_c$, for which we exploit \JD.
	
	First, we ensure that $f_c$ is a function by the axiom $\top\sqsubseteq\exists f_c.\top_\cDom$, where $\top_\cDom$ is expressed as in the proof of \Cref{thm:cd-assertions}. 
	We next give $f_c$ the value $c$ for all elements reachable from the ABox. If $\Omc$ does not contain an ABox, we introduce a fresh individual name $a$. For every individual name $a$ occurring in the ABox, or for that fresh individual name, we	add the feature assertion $f_c(a,c)$.
	By \JD, equality between two variables $x$, $y$ can be expressed using a quantifier-free formula $\phi(x,y)$ over the signature of $\cDom$ (i.e.\ not including the additional singleton predicates). By \JEPD and finiteness of the signature, we can express negated atoms as disjunctions of positive atoms, so that we may assume $\phi(x,y)$ to be in DNF. Now consider the formula $\phi(c,y)$, where $x$ is replaced by the constant~$c$. Since $\phi(c,y)$ is equivalent to $c=y$, we can find a single disjunct $\alpha(c,y)$ of $\phi(c,y)$ such that $\alpha(c,y)$ is satisfiable and equivalent to $c=y$; otherwise, $\phi(c,y)$ would be satisfied also for values of~$y$ other than~$c$. For every $r\in\NR(\Omc)$, we now obtain the concept~$C_{c,r}$ from $\alpha(c,y)$ by replacing $\wedge$ with $\sqcap$ and every atom $P(t_1,\ldots,t_n)$ with $\forall p_1,\ldots,p_n.P$, where $p_i=f_c$ if $t_i=c$, and $p_i=rf_c$ if $t_i=y$, and add the axiom $\top\sqsubseteq C_{c,r}$ to the TBox. This ensures that, in every model $\Imc$ of the resulting ontology, for all elements $d$ that are reachable from the ABox, we have that $f_c^\Imc(d)=c$. We do not need to consider other elements, since one can show similarly to the proofs of \Cref{lem:soundness,lem:completeness} that, if $\Omc$ is consistent, then it has a forest-shaped model in which every element is reachable from some element in the ABox.
	
	We can now replace every concept of the form $\exists f.{=_c}$ with a concept $C_{c,f}$ that is obtained from $\alpha(c,y)$ by replacing $\land$ with $\sqcap$ and atoms $P(t_1,\ldots,t_n)$ with $\exists f_1,\ldots,f_n.P$, where $f_i=f_c$ if $t_i=c$, and $f_i=f$ if $t_i=y$. 
	Doing this transformation for all singleton predicates~$=_c$ in~$\Omc$ leads to the required equi-consistent ontology~$\Omc'$ without singleton predicates.
\end{proof}

\end{document}